\documentclass[runningheads]{llncs}

\usepackage[english]{babel}
\usepackage[utf8]{inputenc}
\usepackage[T1]{fontenc}
\usepackage{graphicx}

\usepackage{enumitem}
\usepackage{amsmath,amsthm,amssymb}
\usepackage{mathtools}
\usepackage{mathrsfs}
\usepackage{fca}
\usepackage{dsfont} 
\usepackage{cancel}

\usepackage{tikz} 
\usetikzlibrary{shapes,arrows,positioning,calc}

\usepackage[colorlinks,citecolor=blue,urlcolor=black,hidelinks,linktocpage]{hyperref}
\usepackage[all]{hypcap}
\usepackage{cleveref}
\let\cref\Cref

\renewcommand{\epsilon}{\varepsilon}
\renewcommand{\phi}{\varphi}

\newcommand{\K}{\mathbb{K}}
\newcommand{\N}{\mathbb{N}}

\newcommand{\BB}{\mathfrak{B}}
\usepackage[shrink=25,babel,kerning=true]{microtype}
\usepackage{csquotes}
\usepackage{booktabs}
\usepackage[subtle]{savetrees}

\newcommand\blfootnote[1]{%
  \begingroup
  \renewcommand\thefootnote{}\footnote{#1}%
  \addtocounter{footnote}{-1}%
  \endgroup
}

\begin{document}

\title{Boolean Substructures in Formal Concept Analysis}

\author{Maren Koyda\inst{1,2}\orcidID{0000-0002-8903-6960} \and Gerd Stumme\inst{1,2}\orcidID{0000-0002-0570-7908}}

\date{\today}

\institute{%
  Knowledge \& Data Engineering Group,
  University of Kassel, Germany\\[0.5ex]
  \and
  Interdisciplinary Research Center for Information System Design\\
  University of Kassel, Germany\\[0.5ex]
  \email{koyda@cs.uni-kassel.de, stumme@cs.uni-kassel.de}
}
\maketitle

\blfootnote{Authors are given in alphabetical order.
  No priority in authorship is implied.}

\begin{abstract}
	It is known that a (concept) lattice contains an n-dimensional Boolean suborder if and only if the context contains an n-dimensional contra-nominal scale as subcontext. In this work, we investigate more closely the interplay between the Boolean subcontexts of a given finite context and the Boolean suborders of its concept lattice. To this end, we define mappings from the set of subcontexts of a context to the set of suborders of its concept lattice and vice versa and study their structural properties. In addition, we introduce closed-subcontexts as an extension of closed relations to investigate the set of all sublattices of a given lattice. 
	
\end{abstract}

\section{Introduction}
\label{sec:introduction}

In the field of Formal Concept Analysis (FCA) the basic data structure is a so-called formal context. It consists of a set of objects, a set of attributes, and an incidence relation on those sets representing which object \emph{has} which attribute. Each such context gives rise to concepts which consist of a maximal set of objects that all share the same maximal set of attributes. The concepts, ordered by subset relation, form a complete lattice.

One frequently occurring type of substructure (more precisely: suborder or sub(semi)lattice) of a concept lattice are Boolean algebras. In the formal context, they correspond to subcontexts that are isomorphic to a contranominal scale, i.\,e., a context of type $(\{1,\dots,k\},\{1,\dots,k\},\neq)$. This means in particular the existence of $k$ objects that just differ slightly on $k$ attributes. However, despite of the only slight difference, these Boolean subcontexts are responsible for an exponential growth of the concept lattice~\cite{albano2015}. Such Boolean subcontexts occur in real-world data as well as in randomly generated formal contexts~\cite{felde2019formal}.

In this paper we investigate the connection between the Boolean substructures in the formal context and in its corresponding concept lattice.
Based on closed subrelations of a formal context~\cite{Wille87b}, that provide a method to characterize the complete sublattices of the corresponding concept lattice, we introduce \emph{closed-subcontexts} and present a one-to-one correspondence to all sublattices. 
Through this, we merge the obvious two-step-approach of limiting the lattice to an interval and determining its complete sublattices in one structure.
Since this construction is an -- almost arbitrary and difficult to handle -- mixture of subcontext and subrelation and in addition is not directly specific to the field of Boolean substructures, we investigate the connection between Boolean subcontexts and Boolean sublattices and suborders, respectively, in~\cref{sec:connection} in a direct way without having to manipulate the incidence relation. 
To this end, we lift two well-known order embeddings~\cite{fca-book}  to the level of subcontexts and suborders to find the Boolean suborders corresponding to a Boolean subcontext. In addition, we introduce a construction to generate the Boolean subcontext associated to a given Boolean suborder. We combine these methods to investigate to which degree the join and meet operators of the lattice are respected by those maps.

As our work is triggered by complexity issues in data analysis where only finite sets are considered, \textbf{all statements in this paper are about finite sets and structures only}, unless explicitely stated otherwise.

As for the structure of this paper, in \cref{sec:FCA} we recall some basic notions and give a brief introduction to the approaches our investigations are based on. 
Afterwards, in~\cref{sec:related-work} we give a short overview of previous works applied to the investigation of substructures of formal contexts and concept lattices.
In \cref{sec:substructures} we introduce some notions required for our investigation on Boolean substructures. 
We introduce closed-subcontexts in~\cref{sec:closed_subcontexts} to determine the set of all Boolean sublattices.
Our second approach is presented in \cref{sec:connection} where we use embeddings of Boolean structures in concept lattices and construct the subcontexts associated to Boolean suborders.
In \cref{sec:discussion} we compare both approaches, and discuss the differences and their overlap.
We conclude our work and give an outlook in \cref{sec:conclusion}.

To advanced readers, we recommend proceeding directly to Section 4 and~\cref{fig: Schaubild} as it illustrates the connections investigated in this work.  

\section{Recap on FCA and Notations}
\label{sec:FCA}

\subsection{Foundations}
Following, we recall some basic notions from FCA. 
For a detailed introduction we refer to~\cite{fca-book}. A formal context
is triple $\mathbb{K}\coloneqq(G,M,I)$, where $G$ is the finite \emph{object set}, 
$M$ the finite \emph{attribute set}, and $I\subseteq G\times M$ a binary \emph{incidence relation}.
Instead of writing $(g,m)\in I$ for an object $g\in G$ and an attribute $m\in M$, we also write $gIm$
and say \emph{object $g$ has attribute $m$}.
One kind of formal context is the family of \emph{contranominal scales}, denoted by $\N^c(k)\coloneqq(\{1,2,...,k\},\{1,2,...,k\},\not =)$.

On the power set of the objects and the power set of the attributes there are two operations given:$
\cdot'\colon\mathcal{P}(G)\to\mathcal{P}(M),~A\mapsto A'\coloneqq
\{m\in M\mid \forall g\in A\colon (g,m)\in I\}$ and $ 
\cdot'\colon\mathcal{P}(M)\to\mathcal{P}(G),~ B\mapsto B'\coloneqq\{g\in
G\mid \forall m\in B\colon (g,m)\in I\}
$
Instead of $A'$ we also write $A^I$ to specify which incidence relation is used for the operation.
A \emph{formal concept} $C=(A,B)$ of the context $(G,M,I)$ is a pair consisting of an object subset 
$A\subseteq G$, called \emph{extent}, and an attribute subset $B\subseteq M$, called \emph{intent}, 
that satisfies $A'=B$ and $B'=A$.
An object set $O\subseteq G$ is called \emph{minimal object generator} of a concept $(A,B)$ if $O''=A$ and $P''\not=A$ for every proper subsets $P\subsetneq O$. Analogous, the \emph{minimal attribute generator} of a concept $(A,B)$ is defined. The set of all minimal object generators (or rather all minimal attribute generators) of $(A,B)$ is denoted by $minG_{obj}(A,B)$ ($minG_{att}(A,B)$).
The set of all formal
concepts ($\mathfrak{B}(\mathbb{K})$)
together with the order defined by $(A_1,B_1)\leq(A_2,B_2)$ iff $A_1\subseteq A_2$ for two
concepts $(A_1,B_1)$ and $(A_2,B_2)$ determines the \emph{concept lattice}
$\underline{\mathfrak{B}}(\mathbb{K})\coloneqq(\mathfrak{B}(\mathbb{K}),\leq)$. 
The concept lattice of $\N^c(k)$ is called \emph{Boolean lattice of dimension $k$} and is denoted by $\BB(k)\coloneqq\underline{\BB}(\N^c(k))$.

There are two tools for basic structural investigations of a formal context $\mathbb{K}=(G,M,I)$ in FCA. An object $g\in G$  is called \emph{clarifiable} if another object $g\ne h\in G$ with $g'=h'$ exists. Furthermore, an object $g\in G$ is called \emph{reducible} if a set of objects $X\subseteq G$ with $g\not\subseteq X $ and $g'=X'$  exists. Otherwise $g$ is called \emph{irreducible}. The same applies to the set of attributes.
The concept lattice of a context $\K$ that has no clarifiable/ reducible objects and attributes is isomorphic to the lattice of any context that can be constructed by adding reducible or clarifiable objects or attributes to $\K$. The stepwise elimination of all clarifiable/ reducible attributes and objects of a formal context results in a \emph{clarified}/\emph{reduced} context, the \emph{standard context} of $\underline{\BB}(\K)$.

To study particular parts of a formal context the selection of a \emph{subcontext} is useful. A \emph{subcontext} $\mathbb{S}\coloneqq(H,N,J)$ of a formal context $\mathbb{K}=(G,M,I)$
is a formal context with $H\subseteq G$, $N\subseteq M$ and $J= I\cap (H \times N)$.
We write $\mathbb{S}\le\K$ to describe $\mathbb{S}$ as a subcontext of $\K$
and use the notion $[H,N]$ instead of $(H,N,I\cap (H\times N))$.
The set of all subcontexts of a formal context $\mathbb{K}$ is denoted by $\mathcal{S}(\mathbb{K})$.

$1_{\underline{L}}$ and $0_{\underline{L}}$ denote the top and the bottom element of a lattice $\underline{L}$. The elements covering $0_{\underline{L}}$ are called \emph{atoms} and the elements covered by $1_{\underline{L}}$ \emph{coatoms}. We denote by $At(\underline{L})$ and $CoAt(\underline{L})$, respectively, the set of all atoms and coatoms of $\underline{L}$. 
$\underline{S}=(\text{S}, \leq)$, a subset $S\subseteq L$ together with the same order relation as $\underline{L}$, is called \emph{suborder} of $\underline{L}$. The set of all suborders of $\underline{L}$ is denoted by $\mathcal{SO}(\underline{L})$.
If $(a,b\in S \Rightarrow (a\vee b)\in S)$ holds we call $\underline{S}$ \emph{sub-$\vee$-semilattice} of $\underline{L}$. 
If $(a,b\in S \Rightarrow (a\wedge b)\in S)$ holds we call $\underline{S}$ \emph{sub-$\wedge$-semilattice} of $\underline{L}$. 
A $\underline{S}$ that is both, a sub-$\vee$-semilattice and a sub-$\wedge$-semilattice, is called \emph{sublattice} of $\underline{L}$. The set of all sublattices of $\underline{L}$ is denoted by $\mathcal{SL}(\underline{L})$.
If  $(T\subseteq S \Rightarrow (\bigvee T)$, $(\bigwedge T)\in S)$ holds $\forall$ $T\subseteq S$ we call $\underline{S}$ \emph{complete sublattice} of $\underline{L}$. 
The requirement for completeness can be translated into $1_{\underline{L}}$ and $0_{\underline{L}}$ being included in $\underline{S}$ if $\underline{L}$ is a finite lattice.

\subsection{Relating Substructures in FCA}

Wille~\cite{Wille87b} presents closed relations to characterize complete sublattices of a concept lattice. 
A relation $J\subseteq I$ is called \emph{closed relation} of a formal context $\mathbb{K}=(G,M,I)$ if every concept of the context
$(G,M,J)$ is a concept of $\mathbb{K}$ as well. Closed relations are linked to the complete sublattices of $\underline{\mathfrak{B}}(\mathbb{K})$~\cite[chap. 3.3]{fca-book}: The set of all closed subrelations of $\mathbb{K}$ and all complete sublattices of $\underline{\mathfrak{B}}(\mathbb{K})$ have a one-to-one correspondence. The bijection
\(C(\underline{S})\coloneqq \bigcup\{A\times B|(A,B)\in \underline{S}\}\) maps the set of all complete sublattices to the set of all closed relations. By limiting the lattice to an interval, the described one-to-one correspondence can be found between the complete lattices of the interval and the closed relations of the formal context associated to the interval.

A connection of the concept lattices of a formal context $\mathbb{K}=(G,M,I)$ and its subcontext $\mathbb{S}=[H,N]$ is given by Ganter and Wille~\cite[Proposition 32]{fca-book} by the two maps
$
\phi_1:\underline{\mathfrak{B}}[H,N]\rightarrow \underline{\mathfrak{B}}(G,M,I),~
(A,B)\mapsto(A'',A')$ and $
\phi_2:\underline{\mathfrak{B}}[H,N]\rightarrow \underline{\mathfrak{B}}(G,M,I),~
(A,B)\mapsto(B',B'')$.
Both maps are \emph{order embeddings}. This means for all $(A_1,B_1),(A_2,B_2)\in \underline{\BB}[H,N]$ that
 $(A_1,B_1)\le(A_2,B_2)$ in $\underline{\BB}[H,N]$ if and only if $\phi_i(A_1,B_1)\le\phi_i(A_2,B_2)$ in $\underline{\BB}(G,M,I)$ for both $i\in\{1,2\}$.
Hence, every structure contained in $\underline{\mathfrak{B}}(\mathbb{S})$ also appears in $\underline{\mathfrak{B}}(\mathbb{K})$.

\section{Related Work}
\label{sec:related-work}

In the field of Formal Concept Analysis, there are several approaches to 
analyze smaller parts of a formal context or a concept lattice, as well as to investigate the connection between the two data structures.
In~\cite{albano2017rich} local changes to a formal context and their effects on the corresponding concept lattice, namely the number of concepts, are explored. Albano~\cite{albano2017polynomial} studies the impact of contranominal scales in a formal context to the size of the corresponding concept lattice by giving an upper bound for $\BB(k)$-free lattices. 
The approach of Wille~\cite{Wille87b} on the one-to-one correspondence between closed subrelations of a formal context complete sublattices of the associated concept lattice is the basis for the work of Kauer and Krupke~\cite{kauer2017generating}. They investigate the problem of constructing the closed subrelation referring to a complete sublattice generated by a given subset of elements while not computing the whole concept lattice.
Based on granulation as introduced in~\cite{zadeh1997toward} the authors of~\cite{qi2019multi} analyze substructures of formal contexts and concept lattices by considering them as granules that provide different levels of accuracy.

Also, many common methods deal with the detection of substructures in the first place. They are based on the selection of structurally meaningful attributes and objects of a formal context. 
For this purpose, Hanika et al.~\cite{hanika2019relevant} search for a \emph{relevant} attribute set that reflects the original lattice structure and the distribution of the objects as good as possible. 
Considering many-valued contexts, Ganter and Kuznetzov~\cite{ganter2008scale} select features based on their scaling. 
Another approach is to generate a meaningful subset by selecting entire concepts directly of the formal context by measuring their individual value for the context and the associated concept lattice.
A natural idea is the consideration of extent and intent size of the concepts. Based on this, Kuznetsov~\cite{kuzuetsov1990stability} proposed a stability measure for formal concepts, measuring the ratio of extent subsets generating the same intent. 
Another measure, the support, was used by Stumme et al.~\cite{STUMME2002189} to generate so-called \emph{iceberg lattices}, which also have a use in the field of mining of frequent association rules.

Besides meaningful reduction, altering the dataset is a standard method in FCA, which is motivated by an attempt to reduce the complexity of the dataset or deal with noise.
In this realm, Dias and Vierira investigate the replacement of \emph{similar} objects by a single representative~\cite{diasreducing}.
Approximate frequent itemsets have been investigated to handle noisy data~\cite{liu2006mining}, where the authors state an additional threshold for both rows and columns of the dataset.

Since we aim to investigate existing substructures of formal contexts and concept lattices, we turn away from those notions in general.

\section{Boolean Subcontexts and Sublattices}
\label{sec:substructures}

In this work, we investigate Boolean substructures in formal contexts as well as in the corresponding concept lattices. 
Therefore, as illustrated in \cref{fig: Schaubild}, we link the different substructures of a formal context with the substructures of the corresponding concept lattice.
In this section we introduce the concrete definitions that serve as a foundation to analyze those connections.

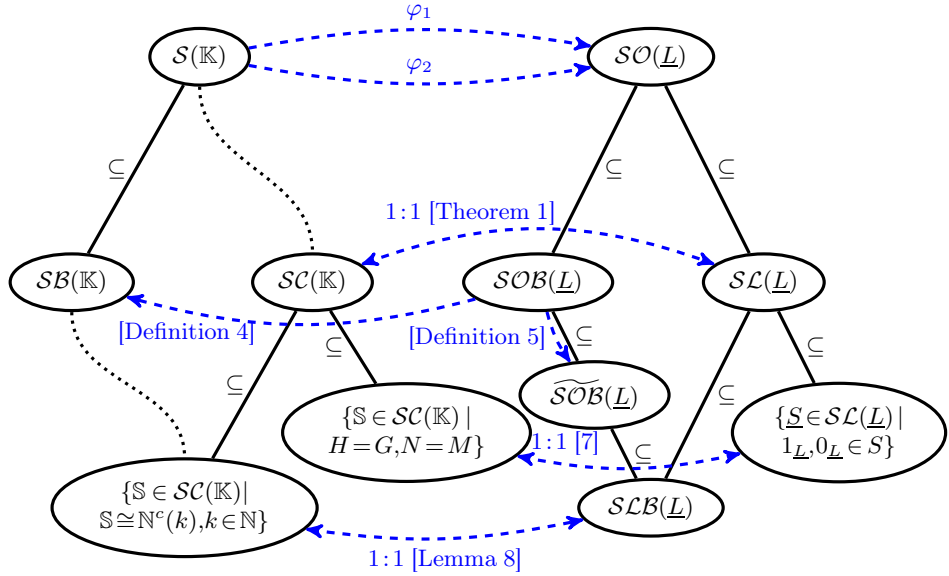
\begin{figure}[t]
	\centering
	\begin{tikzpicture}
	\tikzset{
		>=stealth',
		box/.style={ ellipse, rounded corners, draw=black, very thick,
			minimum height=0.5cm, text centered, node
			distance=5.5cm}, label/.style={ text width=8.0em, text
			centered, },
		pil/.style={ -, very thick}} 
	
	\node
	[box, draw,text width=7em](SCB) at (-0.2,0) {$\{\mathbb{S}\in\mathcal{SC}(\mathbb{K})|$\\$ \mathbb{S} \cong \mathbb{N}^c(k), k\in\mathbb{N}\}$}; 
	\node
	[box, draw](SC) at (1.5,3) {$\mathcal{SC}(\mathbb{K})$};
	\node
	[box, draw](SB) at (-1.7,3) {$\mathcal{SB}(\mathbb{K})$};
	\node
	[box, draw](S) at (0,6) {$\mathcal{S}(\mathbb{K})$};
	\node
	[box, draw, text width=6.5em](SV) at (2.75,1) {$\{\mathbb{S}\in \mathcal{SC}(\mathbb{K})\mid$\\$ H=G, N=M\}$};
	\node
	[box, draw](LCB) at (6,0) {$\mathcal{SLB}(\underline{L})$}; 
	\node
	[box, draw](LC) at (7.5,3) {$\mathcal{SL}(\underline{L})$};
	\node
	[box, draw](LB) at (4.5,3) {$\mathcal{SOB}(\underline{L})$};
	\node
	[box, draw](L) at (6,6) {$\mathcal{SO}(\underline{L})$};
	\node
	[box, draw, text width=5.5em](VL) at (8.5,1) {$\{\underline{S} \in \mathcal{SL}(\underline{L}) \mid$\\$ 1_{\underline{L}},0_{\underline{L}}\in S\}$};
	\node
	[box, draw](SemiL) at (5.25,1.5) {$\mathcal{\widetilde{SOB}}(\underline{L})$};

	\draw[very thick] (S) to
	node[left] {$\subseteq$} (SB);
	
	\draw[very thick] (SC) to
	node[left] {$\subseteq$} (SCB);
	
	\draw[very thick] (SC) to
	node[left] {$\subseteq$} (SV);
	
	\draw[dotted,out=-90, in=90,very thick] (S) to
	node[above] {} (SC);
	
	\draw[dotted,out=-90, in=90,very thick] (SB) to
	node[above] {} (SCB);

	\draw[very thick] (LCB) to
	node[right] {$\subseteq$} (SemiL);

	\draw[very thick] (SemiL) to
	node[right] {$\subseteq$} (LB);
	
	\draw[very thick] (LCB) to
	node[right] {$\subseteq$} (LC);
	
	\draw[very thick] (L) to
	node[right] {$\subseteq$} (LC);
	
	\draw[very thick] (L) to
	node[right] {$\subseteq$} (LB);
	
	\draw[very thick] (LC) to
	node[right] {$\subseteq$} (VL);

	\draw[<->,dashed, blue, very thick] (SC) to[bend left=20]
	node[pos=0.3, above] {$1:1$ [\cref{thm:subcontexts}]} (LC);
	
	\draw[<->,dashed, blue, very thick] (SCB) to[bend right=10]
	node[below] {$1:1$ [\cref{lem:boolean_closed}]} (LCB);
	
	\draw[->,dashed, blue, very thick] (S) to[bend left=10]
	node[above] {$\phi_1$} (L);
	
	\draw[->,dashed, blue, very thick] (S) to[bend right=10]
	node[above] {$\phi_2$} (L);
	
	\draw[<->, blue,dashed, very thick] (SV) to[bend right=12]
	node[pos=.2,above] {$1:1$ \cite{fca-book}}(VL);
	
	\draw[->, blue,dashed, very thick] (LB) to [bend right=10]
	node[left] {[\cref{def:associated}]}(SemiL);
	
	\draw[->, blue,dashed, very thick] (LB) to [bend left=15]
	node[pos=0.84,below] {[\cref{def:associated_context}]}(SB);
	
	\end{tikzpicture}
	\caption{Connections between the subcontexts of a formal context $\mathbb{K}$ and the suborders of the corresponding concept lattice $\underline{L}:=\underline{\mathfrak{B}}(\K)$. The set of all subsemilattices of $\underline{L}$ is denoted by $\mathcal{\widetilde{SOB}}(\underline{L})$.}
	\label{fig: Schaubild}
\end{figure}

\begin{definition}
	Let $\mathbb{K}$ be a formal context, $\mathbb{S}\le \K$. 
	$\mathbb{S}$ is called \emph{Boolean subcontext of dimension k} of $\K$, 
	if $\underline{\BB}(\mathbb{S})\cong\BB(k)$. 
	$\mathbb{S}$ is called \emph{reduced} if $\mathbb{S}$ is a reduced context. 
	The set of all Boolean subcontexts of dimension $k$ of $\mathbb{K}$ and the set of all reduced Boolean subcontexts of dimension $k$ of $\mathbb{K}$ are denoted by $\mathcal{SB}_k(\mathbb{K})$ and $\mathcal{SRB}_k(\mathbb{K})$.
\end{definition}

Note that a reduced Boolean subcontext of dimension $k$ is isomorphic to the contranominal scale $\N^c(k)$.

\begin{definition}
	Let $\underline{L}$ be a lattice and $\underline{S}$ a suborder of $\underline{L}$. 
	$\underline{S}$ is called \emph{Boolean suborder of dimension k} if $\underline{S}\cong \BB(k)$. 
	If $\underline{S}$ is a sublattice of $\underline{L}$, $\underline{S}$ is called \emph{Boolean sublattice of dimension k}.
	The set of all Boolean suborders of dimension $k$ of a lattice $\underline{L}$ is denoted by $\mathcal{SOB}_k(\underline{L})$.
	The set of all Boolean sublattices of dimension $k$ of a lattice $\underline{L}$ is denoted by $\mathcal{SLB}_k(\text{ \underline{L}})$.
\end{definition}

If all dimensions are considered, the number $k$ is left out in the following.

Note that $\mathcal{SLB}_k(\text{\underline{L}})$ is a subset of $\mathcal{SOB}_k(\text{\underline{L}})$ and the standard context of a Boolean lattice $\underline{L}$ of dimension $k$ consists of a formal context $\K\cong\N^c(k)$~\cite[Proposition 12]{fca-book}. Conversely, a formal context $\K$ consisting of a reduced Boolean subcontext of dimension $k$ and an arbitrary number of additional reducible attributes and objects has a corresponding concept lattice $\underline{\BB}(\K)\cong \BB(k)$.

For a better understanding of these structures, we introduce the example given in \cref{runexp}. We will refer back to this illustration throughout the paper.

\begin{example}
	\label{subcontext}
	\label{expl:reduced_boolean}
	$\mathbb{S}=(\{4,5,6\},\{b,c,d,e\},J)$ with $J=I\cap (\{4,5,6\}\times \{b,c,d,e\})$ is a Boolean subcontext of dimension $3$ of the formal context $\mathbb{K}$ given in~\cref{runexp}. $\mathbb{S}$ is not reduced, since $d^{J}=e^{J}$  holds. However, $\mathbb{S}$ includes the reduced Boolean subcontexts $\mathbb{S}_1=[\{4,5,6\},\{b,c,d\}]$ and $\mathbb{S}_2=[\{4,5,6\},\{b,c,e\}]$. The third reduced Boolean subcontext in $\mathbb{K}$ is $\mathbb{S}_3=[\{1,2,3\},\{a,b,c\}]$. The concept lattice of $\K$ in~\cref{runexp} contains 15 Boolean suborders of dimension 3, two of which are also Boolean sublattices.
\end{example}

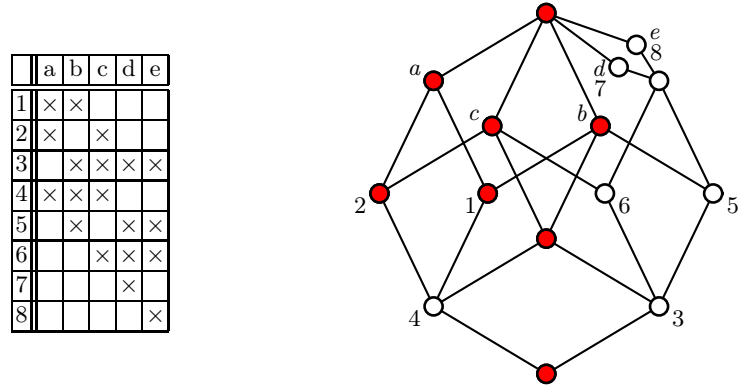
\begin{figure}[t]
	\begin{minipage}{0.4\textwidth}
	\centering
	\begin{cxt}%
		\att{a}%
		\att{b}%
		\att{c}%
		\att{d}%
		\att{e}%
		\obj{xx...}{1} %
		\obj{x.x..}{2} %
		\obj{.xxxx}{3} %
		\obj{xxx..}{4} %
		\obj{.x.xx}{5} %
		\obj{..xxx}{6} %
		\obj{...x.}{7} %
		\obj{....x}{8} %
	\end{cxt}
	\end{minipage}
\begin{minipage}{0.58\textwidth}
	\centering
	
	{\unitlength 0.6mm
		
		\begin{picture}(80,80)%
		\put(0,0){%
						
		\begin{diagram}{80}{80}
		\Node{1}{40}{0}
		\Node{2}{15}{15}
		\Node{3}{65}{15}
		\Node{4}{40}{30}
		\Node{5}{28}{55}
		\Node{6}{52}{55}
		\Node{7}{40}{80}
		\Node{8}{3}{40}
		\Node{9}{77}{40}
		\Node{10}{15}{65}
		\Node{11}{65}{65}		
		\Node{12}{27}{40}
		\Node{13}{53}{40}
		\Node{14}{56}{68}
		\Node{15}{60}{73}
		
		\Edge{1}{2}
		\Edge{1}{3}
		\Edge{3}{4}
		\Edge{2}{4}
		\Edge{4}{5}
		\Edge{4}{6}
		\Edge{7}{5}
		\Edge{7}{6}
		\Edge{8}{5}
		\Edge{8}{2}
		\Edge{9}{6}
		\Edge{9}{3}
		\Edge{10}{7}
		\Edge{8}{10}
		\Edge{9}{11}
		\Edge{12}{2}
		\Edge{12}{6}
		\Edge{12}{10}
		\Edge{13}{3}
		\Edge{13}{5}
		\Edge{13}{11}
		\Edge{11}{14}
		\Edge{11}{15}
		\Edge{14}{7}
		\Edge{15}{7}

		\leftObjbox{2}{3}{1}{4}
		\rightObjbox{3}{3}{1}{3}
		\leftObjbox{8}{3}{1}{2}
		\rightObjbox{9}{3}{1}{5}
		\leftObjbox{12}{3}{1}{1}
		\rightObjbox{13}{3}{1}{6}
		\NoDots\leftObjbox{14}{3}{3}{7}
		\rightObjbox{15}{3}{0}{8}
		
		\leftAttbox{10}{3}{1}{a}
		\leftAttbox{5}{3}{1}{c}
		\leftAttbox{6}{3}{1}{b}
		\NoDots\leftAttbox{14}{3}{-2}{d}
		\rightAttbox{15}{3}{1}{e}
		
		\end{diagram}}
		
		\put(40,0){\ColorNode{red}}
		\put(40,30){\ColorNode{red}}
		\put(28,55){\ColorNode{red}}
		\put(52,55){\ColorNode{red}}
		\put(40,80){\ColorNode{red}}
		\put(3,40){\ColorNode{red}}
		\put(15,65){\ColorNode{red}}
		\put(27,40){\ColorNode{red}}		
		
		\end{picture}}			
\end{minipage}
	\caption{Example of a formal context $\mathbb{K}=(G,M,I)$ with $G=\{1,2,...,8\}$ and $M=\{a,b,...e\}$ containing three reduced Boolean subcontexts
	and its corresponding concept lattice $\underline{\BB}(\K)$.}
	\label{runexp}
\end{figure}

\section{Closed-Subcontexts}
 \label{sec:closed_subcontexts}
At first, we leave the field of (Boolean) suborders and narrow our focus on (Boolean) sublattices. On the context side, we introduce so-called \emph{closed-subcontexts} and show their one-to-one relationship to the sublattices of the concept lattice.

In~\cite{Wille87b}, Wille introduced closed relations of a context to characterize the complete sublattices of its concept lattice. In finite lattices, complete sublattices differ from (non-complete) sublattices in that they always include the top element and the bottom element of the lattice. We adopt Wille's construction to match with (non necessarily complete) sublattices.

\begin{definition}
	\label{def:closed-subcontext}
	Let  $\mathbb{K}=(G,M,I)$ and $\mathbb{S}=(H,N,J)$ be two formal contexts. We call $\mathbb{S}$ \emph{closed-subcontext of $\mathbb{K}$} iff $H\subseteq G$, $N\subseteq M$, $J\subseteq I \cap (H\times N)$ and every concept of $\mathbb{S}$ is a concept of $\mathbb{K}$ as well. 
	The set of all closed-subcontexts of $\mathbb{K}$ is denoted by  $\mathcal{SC}(\mathbb{K})$.
\end{definition}

The sublattices of $\underline{\mathfrak{B}}(\mathbb{K})$ have a one-to-one correspondence to closed-subcontexts of $\mathbb{K}$ as follows.

\begin{theorem}
	\label{thm:subcontexts}
	Let $\mathbb{K}$ be a formal context and 
	$\underline{S}$ be a sublattice of $\underline{\mathfrak{B}}(\mathbb{K})$. 
	Then 
	\[\mathbb{K}_{\underline{S}}\coloneqq (\bigcup_{(A,B)\in \text{\underline{S}}} A, \bigcup_{(A,B)\in \text{\underline{S}}} B, \bigcup_{(A,B)\in \text{\underline{S}}} A\times B)\] 
	is a closed-subcontext of $\mathbb{K}$.	Conversely, for every closed-subcontext $\mathbb{S}$ of $\mathbb{K}$, $\underline{\mathfrak{B}}(\mathbb{S})$ is a sublattice of $\mathfrak{B}(\mathbb{K})$.\\
	Furthermore, the map $f(\underline{S}):=\K_{\underline{S}}$ maps the set of sublattices of $\underline{\BB}(\K)$ bijectively onto the set of closed-subcontexts of $\K$.
\end{theorem}

\begin{proof}
	For each formal concept $(A,B)\in \underline{S}$ the formal concept $(A,B) \in \BB(\mathbb{K}_{\underline{S}})$ is due to construction a concept in $\mathbb{K}$. 
	On the other side let $\mathbb{S}=(H,N,J)$ be a closed-subcontext of $\mathbb{K}$. The concept set of $\mathbb{S}$ is a subset of the concept set of $\mathbb{K}$ and therefore $\underline{\mathfrak{B}}(\mathbb{S})$ is a suborder of $\underline{\mathfrak{B}}(\mathbb{K})$. 
	Let $(A_1,B_1),(A_2,B_2)\in\underline{\BB}(\mathbb{S})$. Let $(A_S,B_S)$ be the infimum of both in $\mathbb{S}$ and $(A_K,B_K)$ the infimum of both in $\mathbb{K}$. So $A_S=A_1\cap A_2=A_K$, which implies $(A_S,B_S)=(A_K,B_K)$ since $(A_S,B_S)$ is by definition a concept in $\K$. The dual argument shows that $\mathbb{S}$ is closed under suprema. So $\underline{\mathfrak{B}}(\mathbb{S})$ is a sublattice of $\underline{\mathfrak{B}}(\mathbb{K})$.
\end{proof}

Note that the closed-subsets of a formal context do not form a closure system since the intersection of two closed-subcontexts, in general, is not a closed-subcontext, even though the sublattices of formal concept do so.

In the construction of $\K_{\underline{S}}$, $\bigcup_{(A,B)\in \underline{S}}A$ is the concept extent of the top element of the sublattice and $\bigcup_{(A,B)\in \underline{S}}B$ is the concept intent of its bottom element.

\begin{lemma}
	Let $\K=(G,M,I)$ be a formal context and $\mathbb{S}=(H,N,J)$ a closed-subcontext of $\K$. 
	Then $H=G$ or $m\in N$ with $m'=H$ exists. And $N=M$ or $ g\in H$ with $g'=N$ exists.
\end{lemma}

\begin{proof}
	Due to~\cref{def:closed-subcontext}, every concept of $\mathbb{S}$ is a concept of $\K$ as well. In particular, this has to hold for the concepts $(\emptyset'',\emptyset')$ and $(H'',H')  $ of $\mathbb{S}$.  
\end{proof}

We provide next some basic statements about closed-subcontexts. Since the following lemmas are based on the work of Wille~\cite{Wille87b} and lifted to our approach, the proofs are similar to the ones in \cite[Section 3.3]{fca-book}.

\begin{lemma} 
	For every set $T\subseteq \mathfrak{B}(G,M,I)$ there is a smallest closed-subcontext $\mathbb{S}$ of $\mathbb{K}$, that contains all $(A\times B)$ for $(A,B)\in T$. $\underline{\BB}(\mathbb{S})$ is the \emph{sublattice of $\underline{\BB}(\K)$ generated by $T$}. 
\end{lemma}
\begin{proof}The proof follows the structure of the proof of Proposition 45 in \cite{fca-book}.\end{proof} 

\begin{lemma}
	$\mathbb{S}=(H,N,J)$ is a closed-subcontext of the formal context $\mathbb{K}=(G,M,I)$ iff $X^{JJ}\supseteq X^{JI}$ holds for each $X\subseteq H$ and for each $X\subseteq N$. 
\end{lemma}
\begin{proof}The proof follows the structure of the proof of Proposition 46 in \cite{fca-book}.\end{proof}

\begin{lemma}
	The closed-subcontexts $(H,N,J)$ of $(G,M,I)$ are exactly the subcontexts that satisfy the condition: 
	(C) If $(g,m)\in (H\times N)$ and $(g,m)\in I\setminus J$ then $(h,m)\not\in I$ for $h\in H$ with $g^{J} \subseteq h^{J}$ and $(g,n)\not\in I$ for $n\in N$ with $m^{J}\subseteq n^{J}$.	
\end{lemma}
\begin{proof}The proof follows the structure of the proof of Proposition 47 in \cite{fca-book}.\end{proof}

\begin{lemma}
	\label{lem:fca-buch_49}
	Let $\mathbb{K}=(G,M,I)$ be a formal context. A clarified formal context $\mathbb{S}=(H,N,J)$ is a closed-subcontext of $\mathbb{K}$ if and only if $H\subseteq G$, $N\subseteq M$ and $J\subseteq I\cap(H\times N)\subseteq H\times N\setminus (\nearrow^J \cup \swarrow^J)$.
\end{lemma}
\begin{proof}The proof follows the structure of the proof of Proposition 49 in \cite{fca-book}.\end{proof}

\begin{lemma}
	Let $\mathbb{K}=(G,M,I)$ be a formal context and $(A,B)$ and $(C,D)$ concepts of $\mathbb{K}$. Then $(A,B, A \times B)$,	$(A,M,I\cap (A\times M))$ and 
	$(G,B,I\cap (G\times B))$ are closed-subcontexts. 
	If $(A,B)\le (C,D)$ also 
	$(C,B,(A\times B \cup C\times D))$ and 
	$(C,B,I\cap (C\times B))$ are closed-subcontexts.
	The corresponding concept lattices are given through
	$\underline{\mathfrak{B}}(A,B, A \times B)=\{(A,B)\}$, 
	$\underline{\mathfrak{B}}(A,M,I\cap (A\times M))=((A,B)]$,
	$\underline{\mathfrak{B}}(G,B,I\cap (G\times B))=[(A,B))$,
	$\underline{\mathfrak{B}}(C,B,(A\times B \cup C\times D))=\{(A,B),(C,D)\}$, and
	$\underline{\mathfrak{B}}(C,B,I\cap (C\times B))=[(A,B),(C,D)]$.
\end{lemma}
\begin{proof}The proof follows the structure of the proof of Proposition 50 in \cite{fca-book}.\end{proof} 

Also, the set of the arrow relations of a closed-subcontext $\mathbb{S}$ is a subset of the set of the arrow relations of the original context $\mathbb{K}$. 

\begin{lemma}
	Let $\mathbb{K}=(G,M,I)$ be a formal context and $\mathbb{S}=(H,N,J)$ a closed-subcontext. Then $\nearrow^J \subseteq \nearrow^I$ and $\swarrow^J \subseteq \swarrow^I$ holds.
\end{lemma}

\begin{proof}
	Let $g\in H, m\in N$ and $g \swarrow^J m$. Assumed $g\not\swarrow^Im$. Then there exists $h\in G$ with $g^I\subseteq h^I$ and $(h,m)\not\in I$. It follows $g^J \subseteq g^{I\cap (G\times H)} \subseteq h^{I\cap (G\times H)} \Rightarrow h\in h^{I\cap (G\times H)} \subseteq g^{JI}=g^{JJ}\subseteq H$ $\Rightarrow g^J\subseteq h^J$. This is a conflict to $g \swarrow^J m$.
\end{proof}

Now we transfer our approach to the field of Boolean substructures. To find all Boolean sublattices (of dimension $k$) in a lattice $\BB(\K)$ the closed-subcontexts of $\K$ that are Boolean subcontexts as well have to be found. Hence,~\cref{thm:subcontexts} can be restricted in the following way:

\begin{lemma}
	\label{lem:boolean_closed}
	Let $\mathbb{K}$ be a formal context.
	$\underline{S}\in\mathcal{SLB}_k(\underline{\mathfrak{B}}(\mathbb{K}))$ iff $\underline{\BB}(\K_{\underline{S}})\cong \BB(k)$ for
	$\mathbb{K}_{\underline{S}}= (\bigcup_{(A,B)\in \text{\underline{S}}} A, \bigcup_{(A,B)\in \text{\underline{S}}} B, \bigcup_{(A,B)\in \text{\underline{S}}} A\times B)$.
\end{lemma}

To directly identify the Boolean closed-subcontexts in a formal context $\K$, the properties of closed-subcontexts can be utilized. Since every concept in $\K$ is either retained or erased but not altered in a closed-subcontext $\mathbb{S}$, the Boolean structure of $\mathbb{S}$ has to be preserved from $\K$.
Every Boolean subcontext $\mathbb{T}=(H,N,J)\in\mathcal{SRB}(\K)$ provides the Boolean structure.
Lifting each concept $(A_{\mathbb{T}},B_{\mathbb{T}})\in\underline{\BB}(\mathbb{T})$ to a
concept $(A_{\K},B_{\K})\in\underline{\BB}(\K)$ with $A_{\mathbb{T}}\subseteq A_{\K}$ and $B_{\mathbb{T}}\subseteq B_{\K}$, generates an extention of the sets $H,N$ and $J$ that provides a Boolean closed-subcontext $\mathbb{S}=(\widetilde{H},\widetilde{N},\widetilde{J})\in\mathcal{SC}(\K)$ as follows:
$\widetilde{H}\coloneqq H\cup\bigcup_{(A_{\mathbb{T}},B_{\mathbb{T}})\in \underline{\BB}(\mathbb{T})} A_{\mathbb{K}}$, 
$\widetilde{N}\coloneqq H\cup\bigcup_{(A_{\mathbb{T}},B_{\mathbb{T}})\in \underline{\BB}(\mathbb{T})} B_{\mathbb{K}}$
and $\widetilde{J}\coloneqq \bigcup_{(A_{\mathbb{T}},B_{\mathbb{T}})\in \underline{\BB}(\mathbb{T})} (A_{\mathbb{K}}\times B_{\mathbb{K}})$.
This approach is represented through the dotted lines in~\cref{fig: Schaubild}.

\section{Connecting Boolean Suborders and Boolean Subcontexts}
\label{sec:connection}

In this section we investigate the relationship between Boolean subcontexts and Boolean suborders. For this purpose, we use the embeddings $\phi_1$ and $\phi_2$ and expand them to the set of Boolean subcontexts. Further, we present a construction to get from a Boolean suborder to a corresponding Boolean subcontext. Both approaches are analyzed with focus on the structural information they transfer and their interplay.

\subsection{Embeddings of Boolean Substructures}
To investigate the connection between Boolean subcontexts $\mathbb{S}$ of a formal context $\mathbb{K}$ and Boolean suborders of $\mathfrak{B}(\mathbb{K})$ we consider embeddings of $\mathfrak{B}(\mathbb{S})$ in $\mathfrak{B}(\mathbb{K})$. Therefore we lift the embeddings $\phi_1$ and $\phi_2$ introduced in \cref{sec:FCA} to the level of subcontexts and suborders:
\begin{align*}
\phi_1:\mathcal{S}(\mathbb{K})&\rightarrow \mathcal{SO}(\underline{\mathfrak{B}}(\mathbb{K})),~
\mathbb{S}\mapsto (\{\phi_1(C)\mid C\in \mathfrak{B}(\mathbb{S})\}, \le)~\text{and}\\
\phi_2:\mathcal{S}(\mathbb{K})&\rightarrow \mathcal{SO}(\underline{\mathfrak{B}}(\mathbb{K})),~
\mathbb{S}\mapsto (\{\phi_2(C)\mid C\in \mathfrak{B}(\mathbb{S})\}, \le).
\end{align*}
From the input (concept or context), it is clear whether the original or the lifted versions of the embeddings $\phi_1$ and $\phi_2$ are used in the following.
We will, in particular, study these mappings for Boolean subcontexts. In this case, an additional structural benefit arises: The images of reduced Boolean subcontexts are sub-$\vee$-semilattice and sub-$\wedge$-semilattices of the original concept lattice:
\begin{lemma}
	\label{lem:subsemilattice}
	Let $\mathbb{K}$ be a formal context, $\mathbb{S}=[H,N]\in \mathcal{SRB}_k(\K)$. Then 
	$\phi_1 (\underline{\mathfrak{B}}(\mathbb{S}))$ is a sub-$\vee$-semilattice of $\underline{\mathfrak{B}}(\mathbb{K})$ and 
	$\phi_2 (\underline{\mathfrak{B}}(\mathbb{S}))$ is a sub-$\wedge$-semilattice of $\underline{\mathfrak{B}}(\mathbb{K})$.
\end{lemma}

\begin{proof}
	Consider $\phi_1$: Let $J\coloneqq I\cap (H\times N)$ and $(A,B)$ and $(C,D)$ be two concepts of $\underline{\BB}(\mathbb{S})$.
	Then $\phi_1(A,B)\vee \phi_1(C,D)=(A'',A') \vee (C'',C')=((A''\cup C'')'',(A'\cap C'))=((A'\cap C')',(A\cup C)')=((A\cup C)'',(A\cup C)')$ and in addition $((A\cup C)'',(A\cup C)')=\phi_1((A\cup C),(B\cap D))=\phi_1((A,B)\vee (C,D))$.
	Since $\mathbb{S}$ is a reduced Boolean context, it includes all possible object combinations as extents so that $E=E^{JJ}$ holds for every $E\subseteq H$. Therefore, in $\BB(\mathbb{S})$ holds $(A,B)\vee(c; D)=((A\cup C)^{JJ},B\cap D)=(A\cup C,B\cap D)$. The procedure for $\phi_2$ is analogous. 
\end{proof}

Note that this conclusion does not hold for Boolean reducible subcontexts, e.g., the formal context given in ~\cref{runexp} and its subcontext $\mathbb{S}=[\{1237\},\{abce\}]$.

The images of the two maps of a reduced Boolean context are in general just a sub-$\vee$-semilattice and a sub-$\wedge$-semilattice, respectively.
Hence, the images of $\phi_1$ and $\phi_2$ have to be identical for $\mathbb{S}\in\mathcal{SRB}_k(\K)$ to generate a lattice. This means $\phi_1(A,B)=(A'',A)=(B',B'')=\phi_2 (A,B)$ has to hold for all $(A,B)\in \mathfrak{B}(\mathbb{S})$.

For every subcontext $\mathbb{S}=(H,N,J)\le\K$ we can differ between the four cases: 
Case 1 with $A'=A^J=B$, $B'=B^J=A$,
case 2 with $A'=A^J=B$, $A=B^J\subset B'$,
case 3 with $B= A^J\subset A'$, $B'=B^J=A$ and 
case 4 with $B= A^J\subset A'$, $A=B^J\subset B'$.
The condition under which $\phi_1(A,B)=\phi_2 (A,B)$ holds is the following:

\begin{lemma}
	Let $\mathbb{K}=(G,M,I)$ be a formal context and $\mathbb{S}\le\mathbb{K}$. $\phi_1(\mathbb{S})=\phi_2(\mathbb{S})$ holds if and only if for all $(A,B)\in \underline{\BB}(\mathbb{S})$ $(A'\setminus B)\times (B'\setminus A)\subseteq I$ holds. If case 1, 2 or 3 holds for all $(A,B)\in \underline{\BB}(\mathbb{S})$, then $\phi_1(\mathbb{S})=\phi_2(\mathbb{S})$ holds directly.
\end{lemma}

\begin{proof}
	For a concept $(A,B)\in \mathfrak{B}(\mathbb{S})$ the identity of both embeddings leads to
	$\phi_1(A,B)=\phi_2(A,B) \Leftrightarrow (A'',A')=(B',B'')=(B',A') \Leftrightarrow (B'\times A')\subseteq I$. 
	This set can be written as
	$B'\times A'=  A\times B~ \cup~ (B'\setminus A) \times B~ \cup~ A  \times  (A'\setminus B) ~ \cup~ (B'\setminus A) \times (A'\setminus B)$. 
	We know 
	$A\times B \subseteq  I$ since $(A,B)\in \mathfrak{B}(\mathbb{S})$ 
	and 
	$A\times A'\subseteq  I$ and $B'\times B \subseteq  I$ by definition of the $\cdot'$ operator. 
	The remaining part equals $(A'\setminus B)\times (B'\setminus A)$. 
	In cases 1 to 3 $(A'',A')=(B',B'')$ holds by construction. 
\end{proof}

\begin{proposition}
	Let $\mathbb{K}=(G,M,I)$ be a formal context and $\mathbb{S}=[H,N]\in \mathcal{SB}_k(\mathbb{K})$. 
	If $H=G$ or $N=M$, then $\phi_1(\mathbb{S})=\phi_2(\mathbb{S})$ holds.
\end{proposition}

However, the relationship between the images of both mappings $\phi_1$ and $\phi_2$ of a specific concept is always (not only in the Boolean case) the same, namely:
\begin{proposition}
	\label{prop:le}
	Let $\K$ be a formal context and $\mathbb{S}\le\mathbb{K}$. Then $\phi_1(A,B)\leq \phi_2(A,B)$ for all $(A,B)\in \BB(\mathbb{S})$.
\end{proposition}

In particular, an interval containing exactly the concepts $(C,D)\in \underline{\mathfrak{B}}(\mathbb{K})$ with $A\subseteq C$ and $B\subseteq D$ exists between $\phi_1(A,B)$ and $\phi_2(A,B)$ with $\phi_1(A,B)$ as its bottom element and $\phi_2(A,B)$ as its top element. In the extreme case, this interval can comprise all of $\underline{\BB}(\K)$, as the following example shows.

\begin{example}
	Let $\mathbb{K}$ be the formal context in~\cref{expl:chain} and $\mathbb{S}=[\{1,2\},\{a,b\}]\le \K$. 
	For the concept $(A,B)=(\{1,2\},\{a,b\})$ of $\mathbb{S}$, $\phi_1(A,B)=(\{1,2\},\{a,b,c,d\})$ and $\phi_2(A,B)=(\{1,2,3,4\},\{a,b\})$ hold. These are the bottom and the top element of the whole concept lattice of $\K$.
\end{example}

\begin{figure}[t]
	\begin{minipage}{0.43\textwidth}
		\centering
		\begin{minipage}{0.3\textwidth}
			\raggedright
			\begin{cxt}%
				\att{a}%
				\att{b}%
				\att{c}%
				\att{d}%
				\obj{xxxx}{1} %
				\obj{xxxx}{2} %
				\obj{xxx.}{3} %
				\obj{xx..}{4} %
			\end{cxt}
		\end{minipage}
		\begin{minipage}{0.6\textwidth}
			\raggedleft
			{\unitlength 0.6mm
				\begin{diagram}{20}{45}
					\Node{1}{15}{10}
					\Node{2}{15}{25}
					\Node{3}{15}{40}
					\Edge{1}{2}
					\Edge{2}{3}
					\leftObjbox{1}{3}{1}{1,2}
					\NoDots\leftObjbox{1}{11}{-3}{$\phi_1(A,B)=$}
					\leftObjbox{2}{3}{1}{3}
					\leftObjbox{3}{3}{1}{4}
					\NoDots\leftObjbox{3}{11}{-3}{$\phi_2(A,B)=$}
					\rightAttbox{1}{3}{1}{d}
					\rightAttbox{2}{3}{1}{c}
					\rightAttbox{3}{3}{1}{a,b}
			\end{diagram}}
		\end{minipage}
		\label{expl:chain}
		\caption{An example of a formal context $\K$ and its subcontext $\mathbb{S}=[\{1,2\},\{a,b\}]=[A,B]$ with $[\phi_1(A,B),\phi_2(A,B)]=\underline{\BB}(\K)$.}
	\end{minipage}
	\hspace{.08\linewidth}
	\begin{minipage}{0.46\textwidth}
		\centering
		\begin{minipage}{0.35\textwidth}
			\centering
			\begin{cxt}%
				\att{a}%
				\att{b}%
				\att{c}%
				\att{d}%
				\att{e}%
				\obj{xx.x.}{1} %
				\obj{x.xx.}{2} %
				\obj{.xx..}{3} %
				\obj{xx.xx}{4} %
				\obj{...x.}{5} %
			\end{cxt}
		\end{minipage}
		\begin{minipage}{0.6\textwidth}
			\raggedleft
			{\unitlength 0.6mm
				
	\begin{picture}(40,85)%
	\put(0,0){%
				\begin{diagram}{40}{85}
					\Node{1}{20}{5}
					\Node{2}{5}{20}
					\Node{3}{5}{35}
					\Node{4}{20}{35}
					\Node{5}{35}{35}
					\Node{6}{5}{50}
					\Node{7}{5}{65}
					\Node{8}{20}{50}
					\Node{9}{35}{50}
					\Node{10}{20}{80}
					\Edge{1}{2}
					\Edge{1}{4}
					\Edge{1}{5}
					\Edge{3}{2}
					\Edge{3}{6}
					\Edge{3}{8}
					\Edge{4}{6}
					\Edge{4}{9}
					\Edge{5}{9}
					\Edge{5}{8}
					\Edge{7}{6}
					\Edge{10}{8}
					\Edge{10}{9}
					\Edge{10}{7}
					\leftObjbox{1}{3}{1}{1,2}
					\leftObjbox{2}{3}{1}{4}
					\leftObjbox{3}{3}{1}{1}
					\leftObjbox{4}{3}{1}{2}
					\leftObjbox{5}{3}{1}{3}
					\leftObjbox{7}{3}{1}{5}
					\leftAttbox{7}{3}{1}{d}
					\leftAttbox{2}{3}{1}{e}
					\rightAttbox{9}{3}{1}{c}
					\rightAttbox{6}{3}{1}{a}
					\rightAttbox{8}{3}{1}{b}
			\end{diagram}}

		\put(5,65){\ColorNode{red}}
		\put(20,80){\ColorNode{red}}
		\put(20,5){\ColorNode{red}}
		\put(20,50){\ColorNode{red}}
		\put(35,50){\ColorNode{red}}		
		\put(35,35){\ColorNode{red}}
		\put(20,35){\ColorNode{red}}
		\put(5,20){\ColorNode{red}}
	\end{picture}}
		
		\end{minipage}
		\label{fig:Anzahl}
		\caption{Example of a formal context $\K$ with $|\mathcal{SRB}_3(\K)|=|\mathcal{SOB}_3(\underline{\BB}(\K))|=4$.}
	\end{minipage}
\end{figure}

This raises the question whether there is a concept lattice where a Boolean suborder exists that can not be obtained by embedding. This is indeed the case also in \cref{runexp}; see, e.g., the Boolean order marked with filled red circles.

An approach to make any Boolean suborder of a (concept) lattice reachable is to expand $\mathbb{K}$ by additional objects and attributes so that every formal concept $C\in \underline{\mathfrak{B}}(\mathbb{K})$ can be generated by one object and by one attribute. 
For a (concept) lattice $\underline{L}$ this is the case with the context $\K=(L,L,\le)$. Here $\underline{S}\in\mathcal{SOB}_k(\underline{L})$ is the image of both $\phi_1(\mathbb{S})$ and $\phi_2(\mathbb{S})$ for the Boolean subcontext $\mathbb{S}=(S,S,\le)$.

Since we are interested in the connections between the existence of Boolean subcontexts on the one hand and the existence of Boolean suborders on the other hand, we observe a first relationship between these sets.

\begin{lemma}
	\label{lem:existens_booleanorder}
	Let $\mathbb{K}$ be a formal context, $\mathcal{SB}_k(\mathbb{K})\not=\emptyset$. Then $\mathcal{SOB}_k(\underline{\mathfrak{B}}(\mathbb{K}))\not=\emptyset$.
\end{lemma}

\begin{proof}
	Let $\mathbb{S}\in \mathcal{SB}_k(\mathbb{K})$. 
	By definition $\underline{\mathfrak{B}}(\mathbb{S})\cong\BB(k)$.
	Since $\phi_1:\underline{\mathfrak{B}}(\mathbb{S})\mapsto\underline{\mathfrak{B}}(\mathbb{K})$ is an order embedding  $\phi_1(\underline{\mathfrak{B}}(\mathbb{S}))$ is a Boolean suborder of dimension $k$ in $\underline{\mathfrak{B}}(\mathbb{K})$.
\end{proof}

In general the images of $\phi_1(\mathbb{S})$ and $\phi_2(\mathbb{S})$ are neither lattices nor semilattices.
However, we know from~\cref{lem:subsemilattice} that if $\mathbb{S}$ is a reduced Boolean subcontext and $\phi_1(\mathfrak{B}(\mathbb{S}))=\phi_2(\mathfrak{B}(\mathbb{S}))$ holds, there exists a Boolean sublattice $\underline{S}$ of the same dimension in $\underline{\BB}(\K)$. 
We can generalize the previous statement as follows:

\begin{lemma}
	\label{lem:anzahl_boolean}
	Let $\mathbb{K}$ be a clarified formal context and $\mathbb{S}_1,\mathbb{S}_2\in \mathcal{SRB}_k(\K)$ with  $\mathbb{S}_1=[H_1,N_1],\mathbb{S}_2=[H_2,N_2]$ and $\mathbb{S}_1\not =\mathbb{S}_2$. 
	If $H_1\not =H_2$, then $\phi_1(\mathbb{S}_1)\not =\phi_1(\mathbb{S}_2)$ holds. 
	If $N_1\not =N_2$, then $\phi_2(\mathbb{S}_1)\not =\phi_2(\mathbb{S}_2)$ holds.
\end{lemma}

\begin{proof}
	Since $\mathbb{S}_1,\mathbb{S}_2\in \mathcal{SRB}_k(\K)$, $|H_1|=|H_2|$ holds. If $H_1\not =H_2$ holds, $ g_1\in H_1$ with $g_1\not \in H_2$ and $g_2\in H_2$ with $g_2\not \in H_1$ exist. Since $\mathbb{S}_1$ and $\mathbb{S}_2$ are reduced and Boolean there is a concept $C_1=(g_1,g_1'') \in \BB(\mathbb{S}_1)$ and a concept $C_2=(g_2,g_2'')\in\mathfrak{B}(\mathbb{S}_2)$. Hence $\mathbb{K}$ is clarified, $\phi_1(C_1)=(g_1'',g_1')\not =(g_2'',g_2')=\phi_1(C_2)$.
	If $N_1\not =N_2$ holds, the analogous procedure can be executed using $\phi_2$.
\end{proof}

Based on this statement, we can assume that the total number of reduced Boolean subcontexts of a formal context $\K$ is a lower bound of the total number of Boolean suborders of $\underline{\BB}(\K)$:

\begin{conjecture}
	Let $\mathbb{K}$ be a clarified formal context with $|\mathcal{SRB}_k(\K)|=n$. Then $|\mathcal{SOB}_k(\underline{\BB}(\K))|\ge n$ holds.
\end{conjecture}

This conjecture can not be proved as straight forward as~\cref{lem:anzahl_boolean} since $\phi_1$ and $\phi_2$ can be identical for some $\mathbb{S}\in\mathcal{SRB}_k(\K)$. In addition not every Boolean suborder is the image of $\phi_1(\mathbb{S})$ or $\phi_2(\mathbb{S})$ for a $\mathbb{S}\in\mathcal{SRB}_k(\K)$. Both phenomena occur in the example given in~Figure 4, where the marked Boolean suborder is not the image of the embedding by $\phi_1$ or $\phi_2$ of any Boolean subcontext contained in the given formal context, although in this case the number of Boolean subcontexts of dimension $3$ and Boolean suborders of dimension $3$ is identical.

\subsection{Subconcepts associated to Suborders}

After investigating mappings of Boolean subcontexts to Boolean suborders, we now analyze the connection between those substructures the other way around. As presented by 
Albano and Chornomaz~\cite[Prop. 1]{albano2015} every formal context $\K$ contains a Boolean subcontext $\mathbb{S}\in\mathcal{SB}_k(\K)$ if $\BB(\K)$ contains a Boolean suborder $\underline{S}\in \mathcal{SOB}_k(\BB(\K))$.
Based on this statement, we introduce a construction to generate a (not necessarily reduced) Boolean subcontext of a formal context based on a Boolean suborder of the corresponding concept lattice.

\begin{definition}
	\label{def:associated_context}
	Let $\K$ be a formal context and $\underline{S}\in \mathcal{SOB}_k(\underline{\BB}(\K))$.
	We call $\psi(\underline{S}):=[H,N]$ with 
	$H\coloneqq\bigcup_{C\in At(\underline{S})} minG_{obj}(C) $ and 
	$N\coloneqq\bigcup_{C\in CoAt(\underline{S})} minG_{att}(C)$
	the \emph{subcontext of $\K$ associated to $\underline{S}$}.
\end{definition}

Indeed the structure arising from the construction given in~\cref{def:associated_context} is a Boolean subcontext of the same dimension as $\underline{S}$:

\begin{lemma}
	Let $\K$ be a formal context, $\underline{S}\in \mathcal{SOB}_k(\underline{\BB}(\K))$ and $\mathbb{S}=[H,N]\coloneqq\psi(\underline{S})$ the subcontext of $\K$ associated to $\underline{S}$. Then $\mathbb{S}\in \mathcal{SB}_k(\K)$.
\end{lemma}

\begin{proof}
	Let $At(\underline{S})=\{A_1,A_2,...,A_k\}$ and $CoAt(\underline{S})=\{C_1,C_2,...,C_k\}$.  Due to the Boolean structure of $\underline{S}$ the atoms can be ordered holding the following condition: $A_i$ is a lower bound for the set $CoAt(\underline{S})\setminus C_i$ for all $1\le i\le k$ and analogous $C_i$ is an upper bound for the set $At(\underline{S})\setminus A_i$ for all $1\le i\le k$. It follows $gIm$ for all $g\in min_{obj}G(A_i)$, $m\in N\setminus minG_{att}(C_i)$ and $g\not I m$ else. 
	So $\mathbb{S}\cong\mathbb{N}^c(k)$.
\end{proof}

In the following, we study the interplay of the mapping $\psi$ from suborders to subcontexts with the mappings $\phi_1$ and $\phi_2$ from subcontexts to suborders.

\begin{lemma}
	Let $\K$ be a formal context and $\mathbb{S}=[H,N]\in\mathcal{SRB}_k(\K)$. 
	Then $\mathbb{S}=\psi(\phi_1(\mathbb{S}))$ iff for all $n\in N$ $(n',n'')\in CoAt(\phi_1(\mathbb{S}))$ holds and	$\mathbb{S}=\psi(\phi_2(\mathbb{S}))$ holds iff for all $h\in H$ $(h'',h')\in At(\phi_2(\mathbb{S}))$ holds.
\end{lemma}

\begin{proof}
	Consider $\phi_1$: Let $\psi(\phi_1(\mathbb{S}))=[\widetilde{H},\widetilde{N}]$, $H=\{h_1,h_2,...,h_k\}$ and $N=\{n_1,n_2,...,n_k\}$. 
	Due to the construction of $\phi_1$ $At(\phi_1(\mathbb{S}))=\{A_1,A_2,...,A_k\}$ with $A_i=(h_i'',h_i')$. Since every $h_i$ is a minimal object generator of an atom of $\phi_1(\mathbb{S})$ $\widetilde{H}=H$ holds.
	Let $CoAt(\phi_1(\mathbb{S}))=\{C_1,C_2,...,C_k\}$.	
	$\widetilde{N}$ consists of the minimal attribute generators of the coatoms of $\phi_1(\mathbb{S})$. Following, $\widetilde{N}=N$ if and only if a renumbering of the coatoms exists so that $C_i=(n_i',n_i'')$ for all $i\in\{1,2,...,k\}$.
	The procedure for $\phi_2$ is analogous.
\end{proof}

\begin{example}
	Let $\K$ be the formal context in Figure 4 and $\mathbb{S}_1=[\{1,2,3\},\{a,b,c\}]$, $\mathbb{S}_2=[\{2,3,4\},\{a,b,c\}]$, $\mathbb{S}_3=[\{1,2,3\},\{b,c,d\}]$ and $\mathbb{S}_4=[\{2,3,4\},\{b,c,d\}])$ its reduced Boolean subcontexts of dimension 3.
	Then $\mathbb{S}_1=\psi(\phi_1(\mathbb{S}_1))=\psi(\phi_2(\mathbb{S}_1))$, 
	$\mathbb{S}_2=\psi(\phi_2(\mathbb{S}_2))$ and
	$\mathbb{S}_3=\psi(\phi_1(\mathbb{S}_3))$ hold.
\end{example}

\begin{lemma}
	\label{lem:6.11}
	Let $\mathbb{K}$ be a formal context, $\underline{S}\in \mathcal{SOB}_k(\underline{\BB}(\K))$,  $\mathbb{S}\coloneqq\psi(\underline{S})$. Let $C\in \underline{S}\setminus \{0_{\underline{S}},1_{\underline{S}}\}$ with either $C$ not being the supremum (in $\underline{\BB}(\K)$) of a subset of $At(\underline{S})$ or $C$ not being the infimum (in $\underline{\BB}(\K)$) of a subset of $CoAt(\underline{S})$.
	Then $(A,B)$ with $A=\bigcup\{minG_{obj}(X)\mid X\in At(\underline{S}), X\le C \}$ and $B=\bigcup\{minG_{att}(X)\mid X\in CoAt(\underline{S}), X\ge C \}$ is a concept of $\mathbb{S}$ with $\phi_1(A,B)\not= \phi_2(A,B)$.
\end{lemma}

\begin{proof}
	According to the construction of $\mathbb{S}$ there is a concept $(A,B)\in \BB(\mathbb{S})$ as stated. If $C$ is not the supremum of a subset of $At(\underline{S})$, 
	especially $A$ does not generate $C$. Therefore $\phi_1(A,B)=(A'',A')< C$, due to the construction of A. Also $\phi_2(A,B)=(B',B'')\ge C$ and consequently $\phi_1(A,B)< \phi_2(A,B)$. Similarly, if $C$ is not the infimum of a subset of $CoAt(\underline{S})$, $\phi_1(A,B)=(A'',A')\le C$, $\phi_2(A,B)=(B',B'')> C$ and $\phi_1(A,B)< \phi_2(A,B)$.
\end{proof}

\begin{lemma}
	\label{lem:5.11}
	Let $\K$ be a formal context, $\underline{S}\in \mathcal{SOB}(\underline{\BB}(\K))$.
	Then $\phi_1(\psi(\underline{S}))$ is a sub-$\vee$-semilattice and $\phi_2(\psi(\underline{S}))$ is a sub-$\wedge$-semilattice of $\underline{\BB}(\K)$.
\end{lemma}	

\begin{proof}
	Let $\mathbb{S}=[H,N]\coloneqq \psi(\underline{S})$. 
	$H$ is the set of all minimal generators of the atoms of $\underline{S}$. Due to the Boolean structure, all concepts in $\K$ that are generated by a subset of $H$ are exactly the supremum of a subset of $At(\mathbb{S})$. Since this generation corresponds to mapping the concepts $C\in \BB(\mathbb{S})$ with $\phi_1$, $\phi_1(\mathbb{S})$ is a sub-$\vee$-semilattice. The second part of the statement is proved similarly.
\end{proof}

\begin{definition}
	\label{def:associated}
	Let $\K$ be a formal context, $\underline{S}\in \mathcal{SOB}_k(\underline{\BB}(\K))$. We call $\phi_1(\psi(\underline{S}))$ the \emph{sub-$\vee$-sublattice of $\underline{\BB}(\K)$ associated to $\underline{S}$} and $\phi_2(\psi(\underline{S}))$ the \emph{sub-$\wedge$-sublattice of $\underline{\BB}(\K)$ associated to $\underline{S}$}.
\end{definition}

The statement in~\cref{lem:5.11} holds especcially for a $\underline{S}$ being a Boolean subsemilattice or a Boolean sublattice of $\underline{\BB}(\K)$ and provides $\phi_1(\psi(\underline{S}))=\underline{S}$ and $\phi_2(\psi(\underline{S}))=\underline{S}$, respectively, as follows.

\begin{lemma}
	\label{lem:6.14}
	Let $\K$ be a formal context and $\underline{S}\in \mathcal{SOB}_k(\underline{\BB}(\K))$. 
	If $\underline{S}$ is a sub-$\vee$-semilattice, $\phi_1(\psi(\underline{S}))=\underline{S}$. If $\underline{S}$ is a sub-$\wedge$-semilattice, $\phi_2(\psi(\underline{S}))=\underline{S}$. 
\end{lemma}

\begin{proof}
	Let $\underline{S}$ be a sub-$\vee$-semilattice and $\mathbb{S}=[H,N]\coloneqq \psi(\underline{S})$. 
	$H$ is the set of minimal generators of the atoms of $\underline{S}$. Due to the Boolean structure all concepts in $\underline{\BB}(\K)$ that are generated by a subset of $H$ are exactly the supremums of a subset of the atoms of $\underline{S}$. Since this generation corresponds to mapping the concepts $C\in \BB(\mathbb{S})$ with $\phi_1$, every image of $\phi_1(C)$ is contained in $\underline{S}$. The second statement is proved similarly.
\end{proof}

\begin{proposition}
	Let $\K$ be a formal context and 
	$\underline{S}\in \mathcal{SLB}_k(\underline{\BB}(\K))$ a sublattice. 
	Then $\phi_1(\psi(\underline{S}))=\phi_2(\psi(\underline{S}))=\underline{S}$.
\end{proposition}

Our research can be concluded in the following theorems. They give an insight into the interplay of $\phi_1,\phi_2$ and $\psi$ and the structural properties they transfer.

\begin{theorem}
	Let $\K$ be a formal context and $\mathbb{S}\in \mathcal{SB}(\K)$. Then:
	\begin{itemize}
		\item[i)] $\psi(\phi_1(\mathbb{S}))=\mathbb{S}$ iff a sub-$\vee$-semilattice $\underline{S}\in \mathcal{SOB}(\underline{\BB}(\K))$ exists with $\psi(\underline{S})=\mathbb{S}$.
		\item[ii)] $\psi(\phi_2(\mathbb{S}))=\mathbb{S}$ iff a sub-$\wedge$-semilattice $\underline{S}\in \mathcal{SOB}(\underline{\BB}(\K))$ exists with $\psi(\underline{S})=\mathbb{S}$.
		\item[iii)] $\psi(\phi_1(\mathbb{S}))=\psi(\phi_2(\mathbb{S}))=\mathbb{S}$ iff a $\underline{S}\in \mathcal{SLB}(\underline{\BB}(\K))$ exists with $\psi(\underline{S})=\mathbb{S}$.
	\end{itemize}
Furthermore, if $\mathbb{S}$ is reduced,
$\phi_1(\mathbb{S})=\phi_1(\psi(\phi_1(\mathbb{S})))$ and $\phi_2(\mathbb{S})=\phi_2(\psi(\phi_2(\mathbb{S})))$.
\end{theorem}

 \begin{proof}
 	Consider i): $(\Rightarrow)$ follows directly from~\cref{lem:5.11} since $\mathbb{S}$ is the subcontext corresponding to the suborder $\phi_1(\mathbb{S})$. $(\Leftarrow)$ is presented in~\cref{lem:6.14}.
 	ii) is proved similarly and iii) follows from the combination of i) and ii).
 	The last statement follows from the combination of~\cref{lem:subsemilattice} and~\cref{lem:6.11}.
 \end{proof}

\begin{theorem}
	Let $\K$ be a formal context and $\underline{S}\in \mathcal{SOB}(\underline{\BB}(\K))$.
	\begin{itemize}
		\item[i)]Then $\phi_1(\psi(\underline{S}))=\underline{S}$ iff $\underline{S}$ is a sub-$\vee$-semilattice.
		\item[ii)]Then $\phi_2(\psi(\underline{S}))=\underline{S}$ iff $\underline{S}$ is a sub-$\wedge$-semilattice.
		\item[iii)]Then $\phi_1(\psi(\underline{S}))=\phi_2(\psi(\underline{S}))=\underline{S}$ iff $\underline{S}$ is a sublattice.
	\end{itemize}
\end{theorem}

\begin{proof}
	Consider i): $(\Rightarrow)$ follows directly from \cref{lem:5.11}. $(\Leftarrow)$ is presented in \cref{lem:6.14},
	ii) is proved similarly, iii) follows from combining i) and ii).
\end{proof}

Altough $\phi_1$ and $\psi$ (or $\phi_2$ and $\psi$) seem to be (dually) adjoint mappings, they are not. E.g., in \cref{fig:gegenbeispiel} consider the subcontexts $\mathbb{S}_1=[\{1,2,3,4\},\{a,b,c\}]$, $\mathbb{S}=[\{1,2,3,4,5\},\{a,b,c\}]$, and $\mathbb{S}_2=[\{1,2,3,4,5,6\},\{a,b,c\}]$. 
It holds $\varphi_1(\mathbb{S}_1)=
\varphi_1(\mathbb{S}_2)=
\varphi_1(\mathbb{S})=
\varphi_2(\mathbb{S}_2)=
\varphi_2(\mathbb{S}_1)$ -- the image is highlighted in the line diagram, and its associated context is $\mathbb{S}$. This shows that $\psi\circ\varphi_1$ is neither monotonic nor anti-monotonic, and the same holds for $\psi\circ\varphi_2$.

\begin{figure}[t]
		\centering
		\begin{minipage}{0.35\textwidth}
			\centering
			\begin{cxt}%
				\att{a}%
				\att{b}%
				\att{c}%
				\att{d}%
				\att{e}%
				\att{f}%
				\obj{xx....}{1} %
				\obj{x.x...}{2} %
				\obj{.xxx..}{3} %
				\obj{.xx.x.}{4} %
				\obj{.xx..x}{5} %
				\obj{.xxxxx}{6} %
			\end{cxt}
		\end{minipage}
		\begin{minipage}{0.6\textwidth}
			\centering
			{\unitlength 0.6mm
				
				\begin{picture}(40,85)%
				\put(0,0){%

				\begin{diagram}{40}{85}
					\Node{1}{20}{5}
					\Node{2}{35}{20}
					\Node{3}{5}{50}
					\Node{4}{20}{50}
					\Node{5}{35}{50}
					\Node{6}{25}{35}
					\Node{7}{5}{65}
					\Node{8}{20}{65}
					\Node{9}{35}{65}
					\Node{10}{20}{80}
					\Node{11}{35}{35}
					\Node{12}{45}{35}
					\Edge{1}{2}
					\Edge{1}{4}
					\Edge{1}{3}
					\Edge{3}{7}
					\Edge{3}{8}
					\Edge{4}{7}
					\Edge{4}{9}
					\Edge{5}{8}
					\Edge{5}{9}
					\Edge{5}{6}
					\Edge{5}{11}
					\Edge{5}{12}
					\Edge{2}{6}
					\Edge{2}{11}
					\Edge{2}{12}
					\Edge{10}{8}
					\Edge{10}{9}
					\Edge{10}{7}
					\rightObjbox{2}{3}{1}{6}
					\leftObjbox{3}{3}{1}{1}
					\leftObjbox{4}{3}{1}{2}
					\rightObjbox{6}{3}{1}{3}
					\rightObjbox{11}{3}{1}{4}
					\rightObjbox{12}{3}{1}{5}
					\leftAttbox{7}{3}{1}{a}
					\rightAttbox{9}{3}{1}{c}
					\rightAttbox{6}{3}{1}{d}
					\rightAttbox{11}{3}{1}{e}
					\rightAttbox{12}{3}{1}{f}
					\rightAttbox{8}{3}{1}{b}
			\end{diagram}}
		
		\put(5,50){\ColorNode{red}}
		\put(5,65){\ColorNode{red}}
		\put(20,65){\ColorNode{red}}
		\put(35,65){\ColorNode{red}}
		\put(20,80){\ColorNode{red}}
		\put(20,5){\ColorNode{red}}
		\put(20,50){\ColorNode{red}}
		\put(35,50){\ColorNode{red}}		
		
	\end{picture}}

		\end{minipage}
		\label{fig:gegenbeispiel}
		\caption{Example of a formal context that shows that neither $\phi_1$ and $\psi$ nor $\phi_2$ and $\psi$ are (dually) adjoint mappings.}
\end{figure}
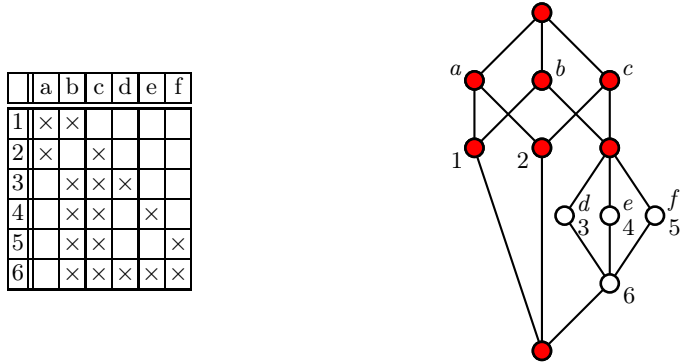

\section{Interplay of both approaches}
\label{sec:discussion}
In the previous sections, two approaches to relate Boolean substructures of a formal context $\mathbb{K}$ with those of the corresponding concept lattice $\mathfrak{B}(\mathbb{K})$ were introduced. In this section, we set both of them in relation.

In \cref{sec:closed_subcontexts} a one-to-one correspondence between the closed-subcontexts of a formal context $\K$ and the sublattices of $\underline{\BB}(\K)$ is presented. However, subsemilattices and suborders are not addressed. In addition, the closed-subcontexts restrict not only the object set and the attribute set of a formal context but also its incidence relation, whereby they could be understood as a more substantial altering of $\K$ compared to the approach presented in~\cref{sec:connection}.
It provides different maps to associate specific Boolean suborders on the one side with Boolean subcontexts on the other side while transferring some structural information.

The intersection of both approaches is localised in the Boolean subcontexts that are closed-subcontexts as well and in general the subcontexts $\mathbb{S}\le\mathbb{K}$ with $C\in \BB(\mathbb{K})$ for all $C\in \BB(\mathbb{S})$.

\begin{lemma}
	 Let $\K$ be a formal context. $\mathbb{S}\le \mathbb{K}$ is a closed-subcontext of $\mathbb{K}$ iff $\phi_1(C)=\phi_2(C)=C$ for all $C\in\mathfrak{B}(\mathbb{S})$.
\end{lemma}

This statement can be restricted to Boolean subcontexts. E.g., the Boolean subcontext $\mathbb{S}=[G,\{a,b,c\}]$ in \cref{runexp} fulfils the requirement. In general, the set of the Boolean subcontexts of $\K$ that are closed-subcontexts is smaller than the set of all Boolean sublattices of $\underline{\BB}(\K)$. So not every Boolean sublattice of $\underline{\BB}(\K)$ can be reached by an embedding of a subcontext of such a structure. 
Refering to those structures we expand the statement of \cref{lem:existens_booleanorder} as follows:

\begin{lemma}
	Let $\K$ be a formal context and $\mathbb{S}\in\mathcal{SB}_k(\mathbb{K})$ with $\mathbb{S}$ a closed-subcontext of $\K$. Then $\underline{S}\coloneqq\phi_1(\mathbb{S})=\phi_2(\mathbb{S})\in \mathcal{SLB}_k(\underline{\BB}(\K))$.
\end{lemma}

However, in general the subcontext $\tilde{\mathbb{S}}$ associated to $\underline{S}$ is not equal to $\mathbb{S}$. E.g. in \cref{runexp} the subcontext $\mathbb{S}=[G,\{a,b,c\}]$ is embedded to a Boolean sublattice $\underline{S}$ but the sublattice, that is associated to $\underline{S}$ is $\tilde{\mathbb{S}}=[\{1,2,3,4\},\{a,b,c\}]$.

\section{Conclusion}
\label{sec:conclusion}

This work relates Boolean substructures in a formal context $\K$ with those in its concept lattice $\underline{\BB}(\K)$.
The notion of closed-subcontexts of $\K$ is presented to generalize closed relations and provide a one-to-one correspondence to the set of all sublattices of $\underline{\BB}(\K)$ using a direct construction. In particular, this relationship can be restricted to the set of all Boolean closed-subcontexts of $\K$, that can be generated based on the set of all reduced Boolean subcontexts of $\K$, and all Boolean sublattices of $\underline{\BB}(\K)$.
Moreover, we investigated two embeddings of Boolean subcontexts of $\K$ into $\underline{\BB}(\K)$. The images of those embeddings are, in general, not sub(semi)lattices but only Boolean suborders and do not cover $\mathcal{SOB}(\K)$ completely. Through the introduction of the subcontext $\mathbb{S}$ associated to a Boolean suborder $\underline{S}$ of $\underline{\BB}(\K)$, the investigated connection is investigated the other way around. The combination of both approaches give an insight of their interplay and the structural information they transfer. Through this every subsemilattice $\underline{S}$ can be associated with a concrete subcontext, that can be mapped to $\underline{S}$ by one of the two embeddings.

We conclude this work with two open questions. First, we are curious to which amount the presented findings can be transferred to general substructures of (not necessarily finite) formal contexts and their corresponding concept lattices. Secondly, we are interested in consideration of other special substructures, e.g., the subcontexts of a concept lattice isomorphic to a nominal scale, as those scales also contain nearly identical objects that differ only in one attribute.

\bibliographystyle{splncs04}
\bibliography{paper_cameraready.bib}

\begin{thebibliography}{10}
\providecommand{\url}[1]{\texttt{#1}}
\providecommand{\urlprefix}{URL }
\providecommand{\doi}[1]{https://doi.org/#1}

\bibitem{albano2017polynomial}
Albano, A.: Polynomial growth of concept lattices, canonical bases and
  generators: extremal set theory in formal concept analysis. Ph.D. thesis,
  SLUB Dresden (2017)

\bibitem{albano2017rich}
Albano, A.: Rich subcontexts. arXiv preprint arXiv:1701.03478  (2017)

\bibitem{albano2015}
Albano, A., Chornomaz, B.: Why concept lattices are large - extremal theory for
  the number of minimal generators and formal concepts. In: International
  Conference on Concept Lattices and Their Applications. {CEUR} Workshop
  Proceedings, vol.~1466, pp. 73--86. CEUR-WS.org (2015)

\bibitem{diasreducing}
Dias, S.M., Vieira, N.: Reducing the size of concept lattices: The {JBOS}
  approach. In: Int. Conf. on Concept Lattices and Their Applications. {CEUR}
  Workshop Proceedings, vol.~672, pp. 80--91. CEUR-WS.org (2010)

\bibitem{felde2019formal}
Felde, M., Hanika, T.: Formal context generation using dirichlet distributions.
  In: ICCS. Lecture Notes in Computer Science, vol. 11530, pp. 57--71. Springer
  (2019)

\bibitem{ganter2008scale}
Ganter, B., Kuznetsov, S.O.: Scale coarsening as feature selection. In: ICFCA.
  Lecture Notes in Computer Science, vol.~4933, pp. 217--228. Springer (2008)

\bibitem{fca-book}
Ganter, B., Wille, R.: Formal Concept Analysis - Mathematical Foundations.
  Springer (1999)

\bibitem{hanika2019relevant}
Hanika, T., Koyda, M., Stumme, G.: Relevant attributes in formal contexts. In:
  ICCS. Lecture Notes in Computer Science, vol. 11530, pp. 102--116. Springer
  (2019)

\bibitem{kauer2017generating}
Kauer, M., Krupka, M.: Generating complete sublattices by methods of formal
  concept analysis. Int. J. Gen. Syst.  \textbf{46}(5),  475--489 (2017)

\bibitem{kuzuetsov1990stability}
Kuznetsov, S.: Stability as an estimate of the degree of substantiation of
  hypotheses derived on the basis of operational similarity. Automatic
  Documentation and Mathematical Linguistics  \textbf{24} (1990)

\bibitem{liu2006mining}
Liu, J., Paulsen, S., Sun, X., Wang, W., Nobel, A.B., Prins, J.F.: Mining
  approximate frequent itemsets in the presence of noise: Algorithm and
  analysis. In: Int. Conf. on Data Mining. pp. 407--418 (2006)

\bibitem{qi2019multi}
Qi, J., Wei, L., Wan, Q.: Multi-level granularity in formal concept analysis.
  Granular Computing  \textbf{4}(3),  351--362 (2019)

\bibitem{STUMME2002189}
Stumme, G., Taouil, R., Bastide, Y., Pasquier, N., Lakhal, L.: Computing
  iceberg concept lattices with titanic. Data \& Knowledge Engineering
  \textbf{42}(2),  189 -- 222 (2002)

\bibitem{Wille87b}
Wille, R.: {Bedeutungen von Begriffsverb{\"a}nden}. In: Beitr{\"a}ge zur
  Begriffsanalyse, pp. 161--211. B.I.--Wissenschaftsverlag, Mannheim (1987)

\bibitem{zadeh1997toward}
Zadeh, L.A.: Toward a theory of fuzzy information granulation and its
  centrality in human reasoning and fuzzy logic. Fuzzy Sets Syst.
  \textbf{90}(2),  111--127 (1997)

\end{thebibliography}

\end{document}